\documentclass[letterpaper, 10pt, conference]{ieeeconf}
\IEEEoverridecommandlockouts                              

\usepackage[utf8]{inputenc}
\usepackage{amsmath, amssymb}
\usepackage{hyperref}
\usepackage{booktabs}
\usepackage{graphicx}
\usepackage{algorithm}
\usepackage{algpseudocode}
\usepackage{tabularx}
\usepackage{filecontents}
\usepackage{tikz}
\usetikzlibrary{positioning, arrows.meta}

\newtheorem{definition}{Definition}[section]

\newtheorem{assumption}{Assumption}[section]

\newtheorem{theorem}{Theorem}[section]
\newtheorem{lemma}{Lemma}[section]
\newtheorem{proposition}{Proposition}[section]

\newtheorem{remark}{Remark}[section]

\title{\LARGE \bf Compiling Spatial Certificates into Temporal Contracts for
  Latency-Aware Control}

\author{Avinash Malik
  \thanks{Avinash Malik is with Department of Electrical, Computer,
    and Software Engineering, University of Auckland, New Zealand {\tt\small avinash.malik@auckland.ac.nz}}%
}

\begin{document}

\maketitle

\begin{abstract}
  We introduce CIPS, a contract-driven execution abstraction for
  managing computational latency and sampled-data updates in
  safety-critical cyber-physical systems (CPS). A fundamental challenge
  in real-time control is that physical safety certificates are defined
  spatially, yet predicting their validity under non-zero computation
  and handoff latency requires online numerical integration of plant
  dynamics. CIPS resolves this operational dichotomy by systematically
  compiling heterogeneous spatial safety certificates into normalized,
  unit-rate temporal contracts entirely offline. This transformation
  abstracts complex plant dynamics, exposing a deterministic, $O(1)$
  temporal budget to a generic sampled-data scheduler. We formally prove
  that this architecture guarantees global hybrid safety invariance
  under bounded computational latency and asynchronous execution.
  Finally, we validate the framework via an autonomous vehicle braking
  benchmark, demonstrating a $60.5\times$ reduction in micro-architectural
  evaluation overhead compared to latency-aware event-triggered control
  (LA-ETC) while preserving safety bounds.
\end{abstract}

\section{Introduction}
\label{sec:intro}

Modern safety-critical Cyber-Physical Systems
(CPS)~\cite{alur2015principles} increasingly rely on formal spatial
certificates, such as Control Barrier Functions
(CBFs)~\cite{ames2019control} and Hamilton-Jacobi reachability
sets~\cite{bansal2017hamilton}, to guarantee operational correctness.
However, continuously evaluating these certificates and synthesizing
updated control laws is computationally prohibitive, necessitating
discrete, sampled-data execution. This introduces a fundamental
divergence between cyber execution and physical dynamics: \textit{a
  control policy and its associated certificate become computationally
  stale during the latency period required to generate a new input, even
  while the physical plant remains safe}.

Consider a practical motivating scenario: a controller is certified safe
for a vehicle state during Adaptive Cruise Control (ACC), but its
certificate is only periodically recomputed. A regeneration request,
such as computing a new control law, may take 20 ms to compute and
verify, while the current certificate has only 19 ms of certified
remaining validity. A standard naive scheduler that knows only that the
certificate is currently valid as a Boolean fact cannot determine
whether regeneration can complete safely before validity expires,
leading to safety violations. Conversely, an advanced scheduler could
predict this deadline dynamically by numerically integrating the plant's
future states online (e.g., Latency-Aware ETC~\cite{tabuada2007event}).
However, this not only introduces a severe runtime computational
bottleneck, but it also fundamentally assumes the scheduler has access
to a perfectly specified, continuous plant model at runtime, information
that is often underspecified or entirely unavailable to the scheduling
middleware. CIPS resolves this dichotomy. It exposes this missing
quantity as a consumable temporal persistence budget, eliminating the
need for online plant simulations by relying on contracts compiled
entirely offline.

CIPS provides a rigorously defined runtime scheduling abstraction to
manage this information staleness. CIPS does not model temporal
persistence of the information representation itself. Rather, it models
the persistence of the certificate's validity under continued execution.
CIPS does not replace existing control theories; rather, CIPS acts as an
information interface. CIPS does not make the plant model unnecessary
for certification; it makes the plant model unnecessary for the runtime
scheduling decision once a sound temporal persistence contract has been
compiled offline.

The central novelty of this work is the architectural pipeline:
transforming an offline spatial certificate into a normalized temporal
persistence contract, which is then consumed by a
plant-model-independent runtime scheduler. We explicitly make the
following three contributions spanning from the offline certificate to
global hybrid invariance:
\begin{enumerate}
\item \textbf{Persistence Compilation:} A formal compiler mapping
  certificate degradation inequalities into normalized unit-rate
  temporal persistence contracts.
\item \textbf{Contract-Based Scheduling:} A
  plant-model-independent sampled runtime scheduler that
  explicitly accounts for regeneration latency, asynchronous completion,
  rejection, and terminal fallback.
\item \textbf{Safety Composition Theorem:} A global hybrid invariance
  theorem showing that the above contracts compose under explicit
  handoff and offline fallback compatibility conditions.
\end{enumerate}

\begin{table}[htbp]
  \centering
  \caption{CIPS Information Flow and Hierarchical Proof Dependencies}
  \label{tab:cips_hierarchy}
  \scriptsize
  \begin{tabular}{@{}ll@{}}
    \hline
    \textbf{Hierarchy Stage} & \textbf{Data Transformation} \\
    \hline
    \textbf{1. Certificate Compiler} & \textbf{In:} Spatial $h_I(x)$, Decay $\alpha(s)$ \\
                                     & \textbf{Out:} Temporal Contract $\Phi_I(x)$ \\[1.2ex]
    \textbf{2. CIPS Contract}        & \textbf{In:} Temporal Contract $\Phi_I(x)$ \\
                                     & \textbf{Out:} Valid Bound $\Phi \le \tau^*$ \\[1.2ex]
    \textbf{3. Latency Safety}       & \textbf{In:} Bound $\Phi \le \tau^*$, Limits ($L_{\text{reg}}, \Delta t$) \\
                                     & \textbf{Out:} Safe Handoff or Fallback \\[1.2ex]
    \textbf{4. Global Hybrid Safety} & \textbf{In:} Safe Handoff or Fallback \\
                                     & \textbf{Out:} Certified System Safety \\
    \hline
  \end{tabular}
\end{table}

Table~\ref{tab:cips_hierarchy} outlines the hierarchical information
flow and mathematical proof dependencies within the CIPS framework. It
illustrates how continuous physical dynamics and spatial bounds are
systematically abstracted into discrete temporal guarantees, culminating
in global system safety. The pipeline begins with the
\textbf{Certificate Compiler} (Stage 1), detailed in
Section~\ref{sec:compilation}, which maps a physical spatial certificate
$h_I(x)$ and its continuous decay bound $\alpha(s)$ into a normalized
temporal contract $\Phi_I(x)$. This functional forms the basis of the
\textbf{CIPS Contract} (Stage 2), formalized in Section
\ref{sec:contract}, establishing that the certified persistence acts as
a valid, conservative lower bound to the exact physical horizon
($\Phi_I(x) \le \tau^*$). By shifting the abstraction from spatial states to
temporal budgets, the framework enables \textbf{Latency Safety} (Stage
3). As analyzed in Sections~\ref{sec:runtime} and~\ref{sec:scheduling},
the runtime scheduler consumes this temporal bound alongside discrete
system limits ($L_{\text{reg}}, \Delta t$) to evaluate trigger margins and
guarantee safe handoffs or terminal fallbacks. Finally, Stage 4 achieves
\textbf{Global Hybrid Safety} by composing these latency-aware
transitions into an indefinite system-wide safety invariant, which is
mathematically proved in Section~\ref{sec:scheduling} and operationally
managed by the sampled event-triggered engine in
Section~\ref{sec:runtime}.

\section{The CIPS Contract and Information Objects}
\label{sec:contract}

We first define the abstract interface that any valid CIPS
implementation must satisfy. Each closed-loop vector field $f_I$ is
assumed to admit a unique forward continuous trajectory.

\begin{definition}[Information Object and Induced Trajectory]
  An \textit{information object} $I_k \in \mathcal{I}$ encapsulates a specific set
  of parameters and control policies. While $I_k$ is actively executing,
  the physical plant evolves according to $\dot{x} = f_{I_k}(x)$. We
  denote the trajectory obtained by continuing the immutable object $I$
  from state $x_0$ at time $t_0$ as $x_I(t_0+\tau; x_0)$.
\end{definition}

\begin{definition}[Physical Safety Horizon]
  Each object $I$ is associated with a ground-truth physical admissible
  region $\mathcal{A}^*_I \subseteq \mathcal{X}$. For a state $x$ evaluated at time
  $t_0$, the \textbf{exact} physical safety horizon is:
\begin{equation}
\tau^*(x; I) := \inf\{\tau \ge 0 \mid x_I(t_0+\tau; x) \notin \mathcal{A}^*_I\}.
\end{equation}
For any boundary or external state $x \notin \mathcal{A}^*_I$, $\tau^*(x; I) = 0$.
\end{definition}

Because computing $\tau^*$ online is generally unavailable to the runtime
scheduler, CIPS relies on a one-sided, \textbf{conservative}
\textit{certified remaining-time lower bound}, termed the persistence
functional.

\begin{definition}[Certified Persistence Functional]
\label{def:cips_contract}
A valid CIPS system provides a computable persistence functional $\Phi_I : \mathcal{X} \to \mathbb{R}$ with $\Phi_I \in C(\mathcal{X})$, which intrinsically defines its certified persistence domain as $\mathcal{D}_I := \{x \in \mathcal{X} : \Phi_I(x) \ge 0\}$. This functional satisfies the following two axioms for all active objects $I$:
\begin{itemize}
    \item \textbf{Axiom 1 (Certificate Soundness):} $\Phi_I(x) \ge 0 \implies x \in \mathcal{A}^*_I$.
    \item \textbf{Axiom 2 (Temporal Persistence):} For every trajectory starting from $x(t_0) \in \mathcal{D}_I$,
    \begin{equation}
    \Phi_I(x(t_0+s)) \ge \Phi_I(x(t_0)) - s
    \end{equation}
    for every $s \in [0, t_e]$ up to the first exit time $t_e$ from $\mathcal{D}_I$.
\end{itemize}
\end{definition}

\begin{remark}
All temporal inequalities are understood only while the trajectory remains inside the certified persistence domain.
\end{remark}

\begin{lemma}[Certified Horizon]
\label{thm:hierarchy}
If $\Phi_I$ satisfies the CIPS Contract, then $0 \le \Phi_I(x) \le \tau^*(x; I)$ whenever $\Phi_I(x) \ge 0$. 
\end{lemma}
\begin{proof}
  Let $\Phi_I(x(t_0)) = r \ge 0$. By Axiom 2, the trajectory satisfies
  $\Phi_I(x(t_0+s)) \ge \Phi_I(x(t_0)) - s = r - s$. Therefore, for every
  $0 \le s < r$, we strictly have $\Phi_I(x(t_0+s)) > 0$. By Axiom 1, this
  implies $x(t_0+s) \in \mathcal{D}_I \subseteq \mathcal{A}^*_I$. Since the state remains strictly
  inside the physically safe set for at least $r$ seconds, the exact
  physical horizon must satisfy $\tau^*(x(t_0); I) \ge r = \Phi_I(x(t_0))$.
\end{proof}

\section{Offline Certificate Compilation}
\label{sec:compilation}

The CIPS contract acts as an abstract specification. The
application-specific proof obligation is to construct a domain
certificate and compile it into a persistence functional $\Phi_I(x)$.

\begin{theorem}[Certificate-to-Persistence Compilation]
  \label{thm:compiler}
  Let $h_I : \mathcal{X} \to \mathbb{R}$ be a continuously differentiable domain certificate
  with a validity threshold $h_{\min}$. Assume explicitly that the domain
  $\mathcal{D}^h_I := \{x \in \mathcal{X} : h_I(x) \ge h_{\min}\}$ satisfies physical soundness:
  $\mathcal{D}^h_I \subseteq \mathcal{A}^*_I$. Assume trajectories $x(t)$ are continuous and there
  exists a function $\alpha \in C([h_{\min}, \infty); \mathbb{R}_{>0})$ such that the
  closed-loop system satisfies $\dot{h}_I(x) \ge -\alpha(h_I(x))$ for all
  $x \in \mathcal{D}^h_I$, and
  $\int_{h_{\min}}^{h_I(x)} \frac{ds}{\alpha(s)} < \infty \quad \forall x \in \mathcal{D}^h_I$. Define the
  functional explicitly as:
  \begin{equation}
    \Phi_I(x) =
    \begin{cases}
      \int_{h_{\min}}^{h_I(x)} \frac{ds}{\alpha(s)}, & h_I(x) \ge h_{\min} \\[2ex]
      -\operatorname{dist}(x, \mathcal{D}^h_I), & h_I(x) < h_{\min}
    \end{cases}
  \end{equation}
  This guarantees that the continuous functional $\Phi_I$ satisfies the CIPS Contract (Definition \ref{def:cips_contract}) on every trajectory up to the first exit from $\mathcal{D}^h_I$.
\end{theorem}
\begin{proof}
  By construction,
  $\Phi_I(x) \ge 0 \iff h_I(x) \ge h_{\min} \implies x \in \mathcal{D}^h_I \subseteq \mathcal{A}^*_I$,
  satisfying Axiom 1. For $x \notin \mathcal{D}^h_I$, the strictly negative distance
  penalization trivially ensures $\Phi_I(x) < 0$, yielding exact
  equivalence of the domains. For Axiom 2, while the trajectory remains
  inside $\mathcal{D}^h_I$, applying the chain rule gives
  $\dot{\Phi}_I(x(t)) = \frac{\dot{h}_I(x(t))}{\alpha(h_I(x(t)))} \ge -1$.
  Integrating this differential inequality from $t_0$ to $t_0 + s$ up to
  the first exit directly yields
  $\Phi_I(x(t_0+s)) - \Phi_I(x(t_0)) \ge -s$, satisfying Axiom 2.
\end{proof}

\begin{remark}
  \label{rem:bhat_bernstein_comparison}
  The integral transform is mathematically related to finite-time
  stability analysis by Bhat and Bernstein~\cite{bhat2000finite}, but
  its role and resulting guarantee are different. Their construction
  transforms a Lyapunov decay inequality into a bound on the time
  required to reach a target equilibrium. Here, the transformation is
  applied to a safety-certificate degradation inequality and yields a
  conservative lower bound on the time remaining before the certificate
  reaches its validity boundary. $\Phi_{I}(x)$ is subsequently exposed as a
  runtime contract, rather than merely being used as an analysis device
  for a particular controller.
\end{remark}

\begin{proposition}[Closure under Finite Intersections]
\label{prop:intersection}
Assume the composite object $I_{\text{comp}}$ induces a common vector
field $f_{I_{\text{comp}}} = f$ under which multiple certificates
$\Phi_i$ satisfy their temporal contracts along the same trajectory. Their
composite object has an admissible region defined as
$\mathcal{A}^*_{\text{comp}} := \bigcap_i \mathcal{A}^*_i$, and the single composite certificate
defined as $\Phi(x) = \min_i \Phi_i(x)$ is also CIPS-valid for
$\mathcal{A}^*_{\text{comp}}$.
\end{proposition}

\section{Runtime Engine and Interface Implementation}
\label{sec:runtime}

The runtime execution induces a hybrid event system modeled by the state
tuple $q = (I, \sigma, x)$. Here, $I$ is the continuous active object,
$\sigma \in \{\text{idle}, \text{pending}, \text{terminal}\}$ is the discrete
scheduler mode, and the physical state flows as $\dot{x} = f_I(x)$.

To ensure plant-model-independence, the scheduler utilizes the runtime
contract
$\mathfrak{C}_I = (\mathcal{O}_I, L_{\text{reg}}, \Delta t, \tau_{\text{op}}, \rho_{\text{reset}},
\mathsf{Verify}, I_{\text{stop}})$. The oracle $\mathcal{O}_I$ evaluates
persistence $\Phi_I(x(t))$ statically from current measurements, avoiding
online numerical simulation. Additionally, $\tau_{\text{op}}(I)$ enforces a
maximum operational age limit, decoupling policy-level deadlines from
certificate-level physical safety.

We define $L_{\text{reg}}$ operationally as the \textit{worst-case
  regeneration trigger-to-certified-handoff latency}. This bound
explicitly encompasses computation, communication, verification, and
handoff time:
$L_{\text{reg}} \ge L_{\text{generation}} + L_{\text{communication}} +
L_{\text{verification}} + L_{\text{handoff}}$.

The scheduler manages updates via sampled evaluations $t_{k}$ (at
intervals $\Delta t$) and asynchronous completions (at times $t_c$),
utilizing two distinct evaluation margins:
\begin{itemize}
\item $\rho_{\text{trigger}} = L_{\text{reg}} + 2\Delta t$: The threshold to
  trigger regeneration (transitions from \textit{idle} to
  \textit{pending}). The $2\Delta t$ factor provides sampling-detection
  conservatism.
\item $\rho_{\text{reset}} > \Delta t$: The minimum persistence required to
  accept a new object (resets from \textit{pending} to \textit{idle}).
\end{itemize}

\begin{algorithm}[htbp]
  \caption{Sampled Event-Triggered CIPS Runtime Engine}
  \label{alg:cips_runtime_engine}
  \footnotesize
  \begin{algorithmic}[1]
    \Require Contract Tuple $\mathfrak{C}_I = (\mathcal{O}_I, L_{\text{reg}}, \Delta t, \tau_{\text{op}}, \rho_{\text{reset}}, \mathsf{Verify}, I_{\text{stop}})$
    \State Initialize: $I_{\text{active}} \gets I_0$, $\text{status} \gets \text{idle}$, $t_{\text{act}} \gets 0$
    \State $\gamma_{\text{last}}^{\text{samp}} \gets \mathcal{O}_{I_0}(t_0)$ 
    \State $\rho_{\text{trigger}} \gets L_{\text{reg}} + 2\Delta t$
    \While{$\text{status} \neq \text{terminal}$}
      \State \textbf{Wait for} Next Sample $t_k$ \textbf{or} Async Completion $t_c$
      
      \If{Event is Sample at $t_k$}
          \State $\gamma_k \gets \mathcal{O}_{I_{\text{active}}}(t_k)$ 
          
          \If{$\gamma_k \le \Delta t$} 
            \State Engage Fallback: Switch controller to $I_{\text{stop}}$
            \State $\text{status} \gets \text{terminal}$
            \State \textbf{break}
          \EndIf
    
          \State $t_{\text{expiry}} \gets t_{\text{act}} + \tau_{\text{op}}(I_{\text{active}})$
          \State $\text{deadline\_cond} \gets (t_k \ge t_{\text{expiry}})$
          
          \If{$\text{deadline\_cond}$ \textbf{and} $\gamma_k \le L_{\text{reg}} + \Delta t$}
            \State Engage Fallback: Switch controller to $I_{\text{stop}}$
            \State $\text{status} \gets \text{terminal}$
            \State \textbf{break}
          \EndIf

          \State $\text{Cross}_k \gets (\gamma_{\text{last}}^{\text{samp}} > \rho_{\text{trigger}} \textbf{ and } \gamma_k \le \rho_{\text{trigger}})$
          \State $\text{DeadlineTrigger}_k \gets (\text{deadline\_cond} \textbf{ and } \gamma_k > L_{\text{reg}} + \Delta t)$
          
          \If{$\text{status} = \text{idle}$ \textbf{and} $(\text{Cross}_k \textbf{ or } \text{DeadlineTrigger}_k)$}
            \State $I_{\text{next}} \gets \text{StartRegeneration}(I_{\text{active}})$
            \State $\text{status} \gets \text{pending}$
          \EndIf
          
          \State $\gamma_{\text{last}}^{\text{samp}} \gets \gamma_k$
      \EndIf

      \If{Event is Async Completion at $t_c$ \textbf{and} $\text{status} = \text{pending}$}
         \State $\gamma_{\text{handoff}} \gets \mathcal{O}_{I_{\text{next}}}(t_c)$ 
         \If{$\mathsf{Verify}(I_{\text{next}})$ \textbf{and} $\gamma_{\text{handoff}} \ge \rho_{\text{reset}}$} 
             \State $I_{\text{active}} \gets I_{\text{next}}$, $t_{\text{act}} \gets t_c$
             \State $\gamma_{\text{last}}^{\text{samp}} \gets \infty$ \Comment{Sentinel to await first sample post-handoff}
         \EndIf
         \State $\text{status} \gets \text{idle}$
      \EndIf
    \EndWhile
  \end{algorithmic}
\end{algorithm}

Algorithm~\ref{alg:cips_runtime_engine} operationalizes the state
machine. Following initialization (lines 1--3), the engine enters an
indefinite event loop, waking upon either an asynchronous completion or
a sampled evaluation (line 5). Sampled events (lines 6--27) explicitly
take precedence if both events co-occur, rigorously managing all
critical safety bounds and regeneration triggers. The engine first
checks for physical unsafety (lines 8--12) or a policy deadline violation
coupled with insufficient buffer (lines 15--19); if either condition is
met, it abandons the active object and breaks into the \textit{terminal}
fallback state. Provided the system remains safe, the scheduler
evaluates the crossing and deadline thresholds (lines 22--25). If a
threshold is crossed while the system is \textit{idle}, it initiates the
generation of a new object and transitions the state to \textit{pending}
(lines 22--25). Conversely, if an asynchronous generation completes while
the system is \textit{pending} (lines 28--35), the oracle evaluates the
new candidate. If the candidate passes verification and meets the reset
margin ($\gamma_{\text{handoff}} \ge \rho_{\text{reset}}$), the object is
successfully handed off, and the system resets to the \textit{idle}
state.

While Algorithm \ref{alg:cips_runtime_engine} guarantees strict safety
by defaulting to a fallback state upon threshold violations, it does not
guarantee regeneration liveness. Delayed or rejected candidates will
simply cause the active object to safely deplete its persistence budget
until fallback is engaged.

\section{Latency-Aware Runtime Scheduling and Global Safety}
\label{sec:scheduling}

Having defined the discrete event-triggered states and hybrid mechanics
of the runtime engine in Section~\ref{sec:runtime}, we now formally
verify its global safety invariants based on the prescribed assumptions.

\begin{assumption}[Runtime Scheduling Semantics]
\label{assum:runtime_semantics}
The event-triggered architecture explicitly enforces the following conditions:
\begin{enumerate}
    \item \textbf{Bounded Latency \& Periodicity:} $L_{\text{reg}} < \infty$ and sampled events occur periodically every $\Delta t > 0$. 
    \item \textbf{Active-Object Immutability:} $I_{\text{active}}(t)$ remains strictly unchanged during a pending regeneration (both while being generated and upon rejection).
    \item \textbf{Handoff State Continuity:} At every nominal transition $I_{\text{active}} \to I_{\text{next}}$, $x(t_{\text{sw}}^+) = x(t_{\text{sw}}^-)$.
    \item \textbf{Verification Soundness:} A candidate $I_{\text{next}}$ is accepted only if its independently verified persistence satisfies $\Phi_{I_{\text{next}}}(x(t_{\text{sw}}^+)) \ge \rho_{\text{reset}}$.
    \item \textbf{Terminal Mode Safety:} The terminal mode $I_{\text{stop}}$ has a domain $\mathcal{D}_{\text{stop}} \subseteq \mathcal{A}^*_{\text{stop}}$ that is forward invariant: $x(t_0) \in \mathcal{D}_{\text{stop}} \implies x(t) \in \mathcal{D}_{\text{stop}} \quad \forall t \ge t_0$. 
    \item \textbf{Fallback Instantaneousness:} Switching to $I_{\text{stop}}$ occurs instantly upon violation of the $\Delta t$ sampled margin.
    \item \textbf{Event Priority (Tie-breaking):} If a sampled event at
      $t_k$ and an asynchronous completion event at $t_c$ occur
      simultaneously ($t_k = t_c$), the sample event strictly takes
      precedence to ensure deterministic hybrid semantic resolution.
\end{enumerate}
\end{assumption}

\begin{assumption}[Offline Fallback Compatibility]
\label{assum:fallback}
Define the reachable fallback switching set for an object $I$ as $\mathcal{S}_{\text{fb}, I} = \{x : 0 < \Phi_I(x) \le \Delta t\}$, and let $\mathcal{S}_{\text{fb}} = \bigcup_I \mathcal{S}_{\text{fb}, I}$. We require $\mathcal{S}_{\text{fb}} \subseteq \mathcal{D}_{\text{stop}}$. Fallback compatibility is intentionally an offline verification obligation. It guarantees that whenever CIPS decides the active certificate is approaching expiry, the current state lies inside the terminal controller's safe domain.
\end{assumption}

\begin{theorem}[Sampled Latency Safety]
\label{thm:latency_safe}
Under Assumption \ref{assum:runtime_semantics}, if regeneration begins
at a sampled instant $t_k$ and
$\Phi_{I_{\text{active}}}(x(t_k)) > L_{\text{reg}} + \Delta t$, then at any asynchronous
completion time $t_c$ satisfying $0 \le t_c - t_k \le L_{\text{reg}}$, the
active object retains at least $\Delta t$ of certified persistence:
$\Phi_{I_{\text{active}}}(x(t_c)) > \Delta t$.
\end{theorem}
\begin{proof}
At $t_c \le t_k + L_{\text{reg}}$, applying Axiom 2 yields $\Phi_{I_{\text{active}}}(x(t_c)) \ge \Phi_{I_{\text{active}}}(x(t_k)) - (t_c - t_k) \ge \Phi_{I_{\text{active}}}(x(t_k)) - L_{\text{reg}} > L_{\text{reg}} + \Delta t - L_{\text{reg}} = \Delta t$.
\end{proof}

\begin{theorem}[Global Hybrid Safety Invariance]
\label{thm:global_safety}
Let $\mathcal{Q} = \{(I, x) : x \in \mathcal{A}^*_I\}$ define the hybrid safety invariant.
This reflects that CIPS does not assume a universal safe set, but rather
a controller-indexed family of safe sets
$\{\mathcal{A}_I^*\}_{I \in \mathcal{I}}$. Suppose every object
$I \in \mathcal{I}$ satisfies the CIPS contract, and assume
$\rho_{\text{reset}} > \Delta t > 0$ with
$\Phi_{I_0}(x_0) \ge \rho_{\text{reset}}$. Under Assumptions
\ref{assum:runtime_semantics} and \ref{assum:fallback}, if the scheduler
executes regeneration upon detecting a valid trigger and invokes
fallback at sample $t_k$ whenever
$\Phi_{I_{\text{active}}}(x(t_k)) \le \Delta t$, then
$(I(t), x(t)) \in \mathcal{Q}$ for all $t \ge 0$.
\end{theorem}
\begin{proof}
We prove $(I(t), x(t)) \in \mathcal{Q}$ holds across all operating modes by establishing the invariant that $x(t) \in \mathcal{A}^*_{I(t)}$ at all times, alongside: \textit{At every sampling instant, either $\Phi_{I_{\text{active}}}(x(t_k)) > \Delta t$ or fallback is immediately invoked.}
\begin{itemize}
    \item \textit{Normal Inter-Sample Execution:} For any sample $t_k$ where fallback is not invoked, $\Phi_{I_{\text{active}}}(x(t_k)) > \Delta t$. Axiom 2 ensures $x(t) \in \mathcal{D}_{I_{\text{active}}} \subseteq \mathcal{A}^*_{I_{\text{active}}}$ for $t \in [t_k, t_{k+1})$.
    \item \textit{Regeneration Pending / Rejected:} Suppose regeneration starts at $t_k$ with $\gamma_k = \Phi_{I_{\text{active}}}(x(t_k)) > \Delta t$. 
    \textbf{Regime A:} If $\gamma_k > L_{\text{reg}} + \Delta t$, then Theorem \ref{thm:latency_safe} guarantees $\Phi_{I_{\text{active}}}(x(t_c)) > \Delta t$ at every completion time $t_c \le t_k + L_{\text{reg}}$. 
    \textbf{Regime B:} If $\Delta t < \gamma_k \le L_{\text{reg}} + \Delta t$, regardless of whether regeneration completes before the next sample, Axiom 2 gives $\Phi_{I_{\text{active}}}(x(t)) > 0$ for $t \in [t_k, t_{k+1}]$ because $t_{k+1} - t_k = \Delta t < \gamma_k$. Hence the active object remains physically safe until the next sample. At $t_{k+1}$, either a successful candidate has been handed off, or the scheduler invokes fallback if the sampled persistence is $\le \Delta t$. In both regimes, safety is preserved.
    \item \textit{Successful Handoff:} State continuity and $\Phi_{I_{\text{next}}}(x(t_{\text{sw}})) \ge \rho_{\text{reset}} > \Delta t$ establish the successor object remains safe for at least the next interval.
    \item \textit{Fallback Invocation:} If $\Phi_{I_{\text{active}}}(x(t_k)) \le \Delta t$ is observed, fallback is immediately invoked. Since $\Phi_{I_{\text{active}}}(x(t_{k-1})) > \Delta t$ at the preceding sample $t_{k-1} = t_k - \Delta t$, Axiom 2 guarantees $\Phi_{I_{\text{active}}}(x(t_k)) \ge \Phi_{I_{\text{active}}}(x(t_{k-1})) - \Delta t > 0$. Thus, immediately prior to fallback at $t_k$, the active object satisfies $0 < \Phi_{I_{\text{active}}}(x(t_k)) \le \Delta t$. This implies $x(t_k) \in \mathcal{S}_{\text{fb}, I_{\text{active}}} \subseteq \mathcal{S}_{\text{fb}}$. By Assumption \ref{assum:fallback}, $\mathcal{S}_{\text{fb}} \subseteq \mathcal{D}_{\text{stop}}$, establishing indefinite safety under $I_{\text{stop}}$.
\end{itemize}
\end{proof}

\section{Illustrative Autonomous Vehicle Certificate \& Experimental
  Validation}
\label{sec:case_study}

\subsection{Experimental Setup}
To conceptually illustrate and validate CIPS, we implement a dynamic
Adaptive Cruise Control (ACC) benchmark operating under bounded
kinematics representing a leader-follower scenario. The physical state
of the system is given by
$x = [p_{\text{lead}}, v_{\text{lead}}, p_{\text{follow}},
v_{\text{follow}}]^T$.

The follower vehicle tracks a leader and must regenerate its control
action $I = u_{\text{follow}}$ safely in the presence of computational
latencies. We enforce strict kinematic constraints: a maximum braking
deceleration $a_b = 6.0 \text{ m/s}^2$, a maximum forward acceleration
$u_{\max} = 2.0 \text{ m/s}^2$, and a maximum allowable velocity
$v_{\max} = 20.0 \text{ m/s}$. Additionally, we mandate a ground-truth
physical minimum boundary $d_{\min} = 3.0 \text{ m}$ representing an
absolute collision threshold, and a measurement uncertainty margin
$\eta = 0.5 \text{ m}$.

\subsection{CIPS Certificate Derivation}
To compile spatial bounds into a temporal persistence contract offline, we define the robust safety threshold $d_{\text{robust}}$:
$d_{\text{robust}} = d_{\min} + \eta + d_{\text{kin}} = 3.0 + 0.5 + 2.05 = 5.55\text{ m}$, 
where $d_{\text{kin}}$ is the kinematic buffer for braking and actuator latency:
$d_{\text{kin}} = v_{\text{rel}} T_{\text{react}} + \frac{v_{\text{rel}}^2}{2 a_b} = (4.0)(0.18) + \frac{4.0^2}{2(6.0)} = 2.05\text{ m}$.

The \textit{spatial certificate} margin is:
$h(x) = (p_{\text{lead}} - p_{\text{follow}}) - d_{\text{robust}}$. 
To avoid open-loop over-conservatism, we bound the decay rate $\alpha(s)$ of $h(x)$. The domain depletion rate is the relative closing speed:
$-\dot{h}(x) = v_{\text{follow}} - v_{\text{lead}} := v_{\text{rel}}$. 
The worst-case relative acceleration closing the gap is:
$a_{\text{rel}}^{\max} = u_{\max} + a_b = 2.0 + 6.0 = 8.0\text{ m/s}^2$.

We explicitly restrict the certified operating domain to states
satisfying the kinematic braking envelope
$v_{\text{rel}}^2 \le v_b^2 + 2 a_{\text{rel}}^{\max} h(x)$, which is
enforced by the offline certificate and controller construction.
Consequently, enforcing a boundary closing-speed $v_b = 4.0\text{ m/s}$
at $h(x) = 0$ yields the kinematic degradation bound:
$-\dot{h}(x) = v_{\text{rel}} \le \alpha(h(x)) = \sqrt{v_b^2 + 2
  a_{\text{rel}}^{\max} h(x)}$

Applying Theorem~\ref{thm:compiler}, the CIPS temporal persistence
functional maps the spatial state into certified validity seconds:
$\Phi(x) = \int_{0}^{h(x)} \frac{ds}{\sqrt{v_b^2 + 2 a_{\text{rel}}^{\max} s}}
= \frac{\sqrt{v_b^2 + 2 a_{\text{rel}}^{\max} h(x)} -
  v_b}{a_{\text{rel}}^{\max}}$

Substituting system parameters yields the closed-form contract:
$\Phi(x) = \frac{\sqrt{16.0 + 16.0 \, h(x)} - 4.0}{8.0}$. This satisfies
$\dot{\Phi}(x) \ge -1$ along all valid trajectories, enabling the runtime
scheduler to evaluate $\Phi(x)$ in $O(1)$ deterministic execution time
without online ODE propagation.

\subsection{Scheduler Implementations}
We evaluated five runtime scheduling policies under latency constraints.
Table~\ref{tab:implementations} outlines their trigger conditions and
evaluation inputs.

\begin{table}[htbp]
  \centering
  \caption{Scheduler Implementations and Trigger Conditions}
  \label{tab:implementations}
  \scriptsize
  \begin{tabular}{l l l}
    \toprule
    \textbf{Scheduler} & \textbf{Evaluation Input} & \textbf{Trigger Condition} \\
    \midrule
    Periodic & Sampling period $dt$ & Time elapsed $\ge P$ \\
    Boolean (Reactive) & Inter-vehicle gap $d(t)$ & $d(t) \le d_{\min}$ \\
    Spatial ETC & Inter-vehicle gap $d(t)$ & $d(t) \le d_{\text{robust}}$ \\
    LA-ETC & Online reachability & $d(t + L_{\text{reg}} + 2\Delta t) \le d_{\text{robust}}$ \\
    CIPS (Proposed) & Persistence $\Phi(x)$ & $\Phi(x) \le L_{\text{reg}} + 2\Delta t$ \\
    \bottomrule
  \end{tabular}
\end{table}

Crucially, LA-ETC has runtime access to the complete plant model $f(x)$
and numerically integrates the closed-loop dynamics forward over
$L_{\text{pred}} = L_{\text{reg}} + 2\Delta t$ at every sample $t_k$ using
explicit 4th-order Runge-Kutta (RK4) integration
($\delta t = 1\text{ ms}$) under worst-case leader braking
($a_{\text{lead}} = -a_b$). In contrast, CIPS performs no forward
propagation and consumes only the precompiled scalar persistence
functional $\Phi_{I}(x)$.

\subsection{Experimental Results \& Evaluation Methodology}

The schedulers were subjected to structured robustness sweeps by varying
both the regeneration latency $L_{\text{reg}}$ (Table~\ref{tab:sweep1})
and the sampled evaluation interval $\Delta t$ (Table~\ref{tab:sweep2}).
Algorithms were compiled in C++-20 with \texttt{-O2} optimizations.

To evaluate runtime overhead, evaluation latency (\texttt{Eval}) was
measured on the Pico-W Cortex-M0+ ARM architecture running upto 133 MHz,
with 16 KB shared data and program L1 cache and 256 KB SRAM. Reported
timing values (Table~\ref{tab:microbenchmark}) reflect the median
execution time per evaluation over $10,000$ iterations.

We evaluate Latency Efficiency ($E_{\text{lat}}$) as the percentage of
the safe temporal budget utilized at the effective regeneration trigger:
$E_{\text{lat}} = \left( 1 -
  \frac{\Phi(x_{\text{trigger}})}{\Phi_{\text{initial}}} \right) \times 100\%$,
where $\Phi_{\text{initial}}$ is the initial persistence budget upon object
activation, and $\Phi(x_{\text{trigger}})$ is the remaining persistence
functional evaluated at the trigger instant. If no regeneration is
required during an evaluation run, the effective trigger is defined at
$\Phi(x_{\text{trigger}})=0$, yielding $E_{\text{lat}}=100\%$.

\begin{table}[htbp]
  \centering
  \caption{Comprehensive Scheduling Metrics \& Robustness Sweep (Varying $L_{\text{reg}}$)}
  \label{tab:sweep1}
  \resizebox{\columnwidth}{!}{%
    \begin{tabular}{lrrrrrrrr}
      \toprule
      Scheduler & L\_reg(ms) & Total Reg & Unnec Reg & Miss/Late & Fallbacks & Min H(m) & Violations & E\_lat (\%) \\
      \midrule
      Periodic & 10 & 625 & 625 & 0 & 0 & 0.74 & 0 & 75.0 \\
      Boolean & 10 & 0 & 0 & 1 & 1 & -12.45 & 1 & 100.0 \\
      Spatial ETC & 10 & 35 & 0 & 1 & 1 & -8.05 & 1 & 98.6 \\
      LA-ETC & 10 & 0 & 0 & 0 & 0 & 0.44 & 0 & 100.0 \\
      CIPS & 10 & 0 & 0 & 0 & 0 & 0.44 & 0 & 100.0 \\
      \midrule
      Periodic & 50 & 454 & 454 & 0 & 0 & 0.74 & 0 & 9.2 \\
      Boolean & 50 & 0 & 0 & 1 & 1 & -12.45 & 1 & 100.0 \\
      Spatial ETC & 50 & 10 & 0 & 1 & 1 & -8.05 & 1 & 98.0 \\
      LA-ETC & 50 & 0 & 0 & 0 & 0 & 0.44 & 0 & 100.0 \\
      CIPS & 50 & 0 & 0 & 0 & 0 & 0.44 & 0 & 100.0 \\
      \midrule
      Periodic & 100 & 238 & 220 & 0 & 0 & 0.74 & 0 & 4.8 \\
      Boolean & 100 & 0 & 0 & 1 & 1 & -12.45 & 1 & 100.0 \\
      Spatial ETC & 100 & 5 & 0 & 1 & 1 & -8.05 & 1 & 98.0 \\
      LA-ETC & 100 & 0 & 0 & 0 & 0 & 0.44 & 0 & 100.0 \\
      CIPS & 100 & 3 & 0 & 0 & 0 & 0.44 & 0 & 98.8 \\
      \bottomrule
    \end{tabular}%
  }
\end{table}

\begin{table}[htbp]
  \centering
  \caption{Robustness as Sampling Interval ($\Delta t$) Varies ($L_{\text{reg}} = 50\text{ms}$)}
  \label{tab:sweep2}
  \resizebox{\columnwidth}{!}{%
    \begin{tabular}{lrrrrrrrr}
      \toprule
      Scheduler & dt(ms) & Total Reg & Unnec Reg & Miss/Late & Fallbacks & Min H(m) & Violations & E\_lat (\%) \\
      \midrule
      Periodic & 5 & 454 & 454 & 0 & 0 & 0.74 & 0 & 9.2 \\
      Boolean & 5 & 0 & 0 & 1 & 1 & -12.45 & 1 & 100.0 \\
      Spatial ETC & 5 & 10 & 0 & 1 & 1 & -8.05 & 1 & 98.0 \\
      LA-ETC & 5 & 0 & 0 & 0 & 0 & 0.44 & 0 & 100.0 \\
      CIPS & 5 & 0 & 0 & 0 & 0 & 0.44 & 0 & 100.0 \\
      \midrule
      Periodic & 20 & 1 & 1 & 0 & 0 & 0.74 & 0 & 99.8 \\
      Boolean & 20 & 0 & 0 & 1 & 1 & -12.45 & 1 & 100.0 \\
      Spatial ETC & 20 & 1 & 0 & 1 & 1 & -8.05 & 1 & 99.8 \\
      LA-ETC & 20 & 0 & 0 & 0 & 0 & 0.44 & 0 & 100.0 \\
      CIPS & 20 & 0 & 0 & 0 & 0 & 0.44 & 0 & 100.0 \\
      \midrule
      Periodic & 50 & 250 & 250 & 0 & 0 & 0.74 & 0 & 50.0 \\
      Boolean & 50 & 0 & 0 & 1 & 1 & -12.45 & 1 & 100.0 \\
      Spatial ETC & 50 & 6 & 0 & 1 & 1 & -8.05 & 1 & 98.8 \\
      LA-ETC & 50 & 0 & 0 & 0 & 0 & 0.44 & 0 & 100.0 \\
      CIPS & 50 & 8 & 0 & 0 & 0 & 0.44 & 0 & 98.4 \\
      \bottomrule
    \end{tabular}%
  }
\end{table}

\begin{table}[htbp]
  \centering
  \caption{Cycle-Accurate Evaluation Micro-benchmark (gem5 ARM)}
  \label{tab:microbenchmark}
  \scriptsize
  \begin{tabular}{lrr}
    \toprule
    Scheduler & Median Eval Time ($\mu$s) & Speedup vs. LA-ETC \\
    \midrule
    CIPS (Ours) & \textbf{2} & \textbf{60.5$\times$} \\
    Spatial ETC & 3 & 40.3$\times$ \\
    Periodic    & 3 & 40.3$\times$ \\
    Boolean     & 3 & 40.3$\times$ \\
    LA-ETC      & 121 & 1.00$\times$ \\
    \bottomrule
  \end{tabular}
\end{table}

\begin{figure}[htbp]
    \centering
    \includegraphics[width=0.9\linewidth]{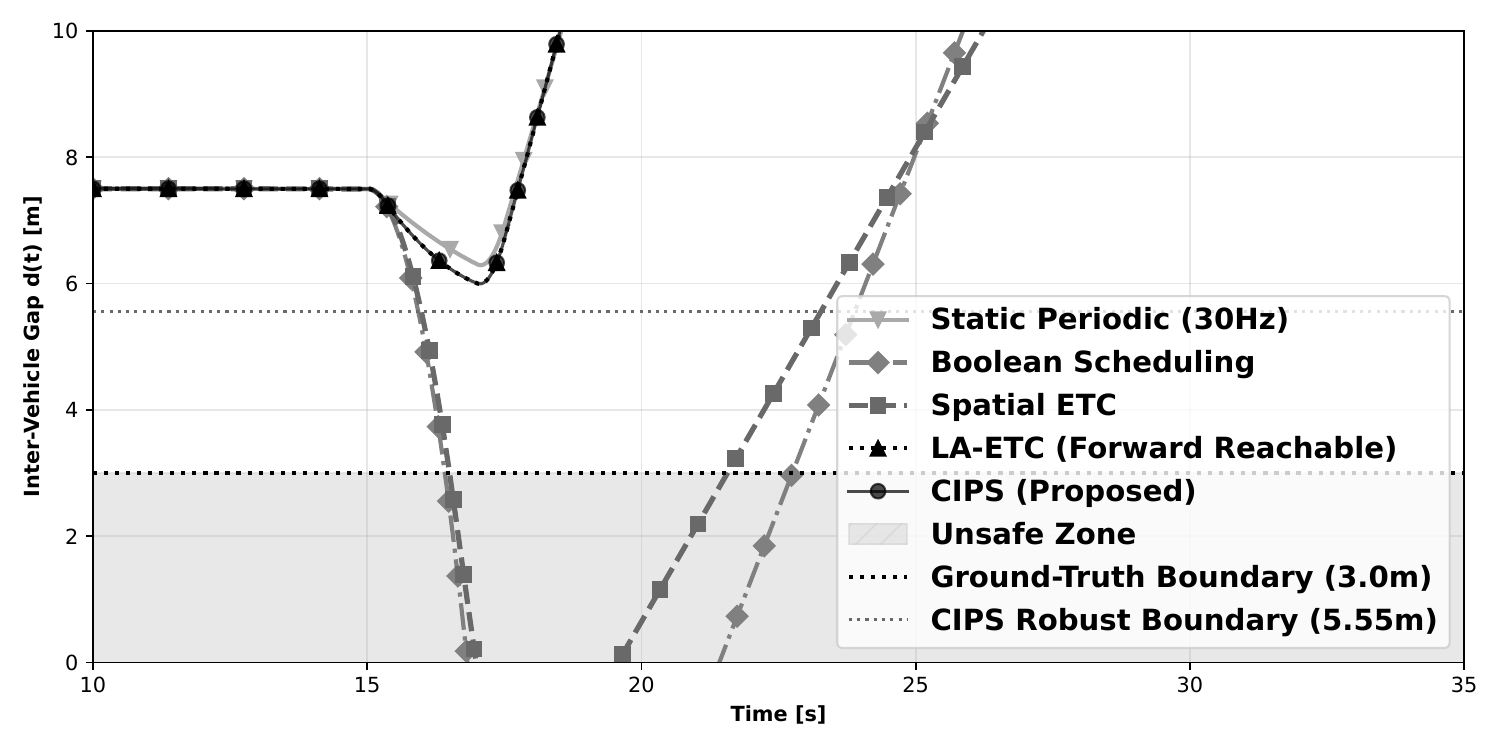}
    \caption{Inter-Vehicle State Validity Tracking illustrating CIPS
      adapting dynamically against Boolean and ETC failures.}
    \label{fig:comparison}
\end{figure}

The benchmark results highlight distinct trade-offs across the
evaluation ladder:
\begin{itemize}
\item \textbf{Periodic Scheduling:} Preserves safety but induces high
  numbers of unnecessary regenerations, degrading latency budget
  efficiency ($E_{\text{lat}}$ drops to 4.8\% at
  $L_{\text{reg}}=100\,\text{ms}$).
\item \textbf{Boolean and Spatial ETC:} Both suffer systemic safety
  violations because they lack temporal latency awareness, triggering
  fallback maneuvers too late (Figure~\ref{fig:comparison}).
\item \textbf{LA-ETC (Online Reachability):} Maintains safety across
  delays by explicitly propagating state trajectories over the latency
  horizon. However, this incurs substantial computational overhead
  ($121\,\mu\text{s}$) due to online numerical integration.
\item \textbf{CIPS Framework:} CIPS exhibits the same safety outcome and
  minimum-gap bound as LA-ETC across evaluated latency and sampling
  settings, while requiring substantially ($60.5\times$) lower evaluation
  overhead. Notably, the results demonstrate substantial latency budget
  efficiency ($E_{\text{lat}}$), remaining near optimal (98.4\% to
  100.0\%). Unlike LA-ETC, whose runtime trigger requires access to and
  numerical propagation of plant dynamics, CIPS consumes only the
  precompiled scalar persistence functional.
\end{itemize}

\section{Related Work and Interface Comparison}
\label{sec:comparison}

Ensuring physical safety under discrete computation and non-zero
latencies is a core cyber-physical systems (CPS) challenge. CIPS
complements existing synthesis and spatial verification techniques by
introducing a temporal compilation layer that normalizes heterogeneous
spatial certificates into consumable runtime contracts.

\textbf{Event-Triggered Control and Predictive Safety Filters.}
Event-Triggered Control (ETC) reduces overhead via state-dependent
updates~\cite{tabuada2007event, heemels2012introduction}, and
latency-aware variants like lookahead ETC (LA-ETC) or predictive safety
filters explicitly incorporate execution delays by evaluating future
trajectories online~\cite{gurriet2018towards, brunke2022safe}. However,
these approaches require continuous online plant vector field
evaluations and numerical integration of plant dynamics, incurring high
computational overhead ($\approx 60.5\times$ more, per
Section~\ref{sec:case_study}). CIPS shifts this burden offline,
eliminating online forward reachability by compiling spatial decay
bounds into a scalar temporal functional $\Phi(x)$. This provides $O(1)$
latency-aware scheduling and drastically reduces deterministic
evaluation overhead.

\textbf{Self-Triggered Safety and Sampled-Data CBFs.} Self-triggered and
sampled-data Control Barrier Functions (CBFs) predict safe sampling
instants or address inter-sampling gaps, but they tie temporal intervals
directly to specific continuous controllers~\cite{anta2010to,
  taylor2022safety, yang2019self}. CIPS decouples this dependency by
compiling any spatial certificate into a platform-agnostic persistence
contract that manages execution budgets and fallback transitions
independently.

\textbf{Reachability-Based Runtime Monitoring.} Hamilton-Jacobi (HJ)
reachability provides exact, non-conservative safety
certificates~\cite{bansal2017hamilton, mitchell2005time}, but evaluating
multi-dimensional implicit value functions online is computationally
prohibitive. CIPS complements offline HJ toolboxes by compiling their
computed spatial domains into fast, microsecond-scale temporal
functionals for runtime execution.

\textbf{Runtime Assurance, Simplex, and Information Staleness.} Runtime
Assurance (RTA) safely overrides primary controllers with verified
backups~\cite{seto1998simplex, sha2001using, cofer2020run}, while Age of
Information (AoI) models data staleness~\cite{kaul2012real,
  yates2021age}. CIPS bridges these paradigms by formalizing staleness
as a consumable temporal asset, yielding an explicit latency-aware
trigger ($\Phi(x) \le L_{\text{reg}} + 2\Delta t$) that guarantees safe mode
switches before physical violations occur.

\section{Conclusion and Future Work}
\label{sec:conclusion}

We introduced CIPS, a contract-driven runtime abstraction that bridges
the gap between continuous physical safety and discrete cyber execution.
By formalizing the Certificate-to-Persistence compiler, we demonstrated
that continuously differentiable spatial certificates can be
systematically transformed into consumable, unit-rate temporal
contracts. Crucially, CIPS shifts the computational burden of plant
model evaluation entirely offline. This decoupling of physical
certificate derivation from runtime scheduling achieves provable global
hybrid safety invariance under bounded latency and sampling, entirely
eliminating the need for online numerical integration.

Future work includes: (1) Extending the compiler for non-smooth barrier
functions and neural network-based bounds. (2) Adapting contracts to
handle probabilistic safety bounds under process noise and environmental
uncertainty. (3) Using advanced reachability bounds to reduce spatial
decay conservatism, yielding tighter, more efficient temporal contracts.

\bibliographystyle{IEEEtran}
\bibliography{references}

@inproceedings{gurriet2018towards,
  title={Towards a framework for realizable safety critical control through active set invariance},
  author={Gurriet, Thomas and Singletary, Andrew and Reher, Jacob and Babuska, Robert and Ames, Aaron D},
  booktitle={2018 ACM/IEEE 9th International Conference on Cyber-Physical Systems (ICCPS)},
  pages={98--106},
  year={2018},
  organization={IEEE}
}

@article{tabuada2007event,
  title={Event-triggered real-time scheduling of stabilizing control tasks},
  author={Tabuada, Paulo},
  journal={IEEE Transactions on Automatic control},
  volume={52},
  number={9},
  pages={1680--1685},
  year={2007},
  publisher={IEEE}
}

@inproceedings{heemels2012introduction,
  title={An introduction to event-triggered and self-triggered control},
  author={Heemels, WPMH and Johansson, Karl Henrik and Tabuada, Paulo},
  booktitle={{IEEE Conference on Decision and Control (CDC)}},
  pages={3270--3285},
  year={2012}
}

@article{anta2010to,
  title={To sample or not to sample: Self-triggered control for nonlinear systems},
  author={Anta, Adolfo and Tabuada, Paulo},
  journal={IEEE Transactions on Automatic Control},
  volume={55},
  number={9},
  pages={2030--2042},
  year={2010},
  publisher={IEEE}
}

@inproceedings{taylor2022safety,
  title={Safety of sampled-data systems with control barrier functions via approximate discrete time models},
  author={Taylor, Andrew J and Dorobantu, Victor D and Cosner, Ryan K and Yue, Yisong and Ames, Aaron D},
  booktitle={2022 IEEE 61st Conference on Decision and Control (CDC)},
  pages={7127--7134},
  year={2022},
  organization={IEEE}
}

@inproceedings{yang2019self,
  title={Self-triggered control for safety critical systems using control barrier functions},
  author={Yang, Guang and Belta, Calin and Tron, Roberto},
  booktitle={2019 American control conference (ACC)},
  pages={4454--4459},
  year={2019},
  organization={IEEE}
}

@inproceedings{bansal2017hamilton,
  title={Hamilton-jacobi reachability: A brief overview and recent advances},
  author={Bansal, Somil and Chen, Mo and Herbert, Sylvia and Tomlin, Claire J},
  booktitle={2017 IEEE 56th annual conference on decision and control (CDC)},
  pages={2242--2253},
  year={2017},
  organization={IEEE}
}

@article{brunke2022safe,
  title={Safe learning in robotics: From learning-based control to safe reinforcement learning},
  author={Brunke, Lukas and Greeff, Melissa and Hall, Adam W and Yuan, Zhaocong and Zhou, Siqi and Panerati, Jacopo and Schoellig, Angela P},
  journal={Annual Review of Control, Robotics, and Autonomous Systems},
  volume={5},
  number={1},
  pages={411--444},
  year={2022},
  publisher={Annual Reviews}
}

@article{mitchell2005time,
  title={A time-dependent Hamilton-Jacobi formulation of reachable sets for continuous dynamic games},
  author={Mitchell, Ian M and Bayen, Alexandre M and Tomlin, Claire J},
  journal={IEEE Transactions on Automatic Control},
  volume={50},
  number={7},
  pages={947--957},
  year={2005},
  publisher={IEEE}
}

@article{sha2001using,
  title={Using simplicity to control complexity},
  author={Sha, Lui},
  journal={IEEE Software},
  volume={18},
  number={4},
  pages={20},
  year={2001},
  publisher={IEEE Computer Society}
}

@inproceedings{cofer2020run,
  title={Run-time assurance for learning-enabled systems},
  author={Cofer, Darren and Amundson, Isaac and Sattigeri, Ramachandra and Passi, Arjun and Boggs, Christopher and Smith, Eric and Gilham, Limei and Byun, Taejoon and Rayadurgam, Sanjai},
  booktitle={NASA Formal Methods Symposium},
  pages={361--368},
  year={2020},
  organization={Springer}
}

@inproceedings{seto1998simplex,
  title={The simplex architecture for safe online control system upgrades},
  author={Seto, Danbing and Krogh, Bruce and Sha, Lui and Chutinan, Alongkrit},
  booktitle={Proceedings of the 1998 American Control Conference. ACC (IEEE Cat. No. 98CH36207)},
  volume={6},
  pages={3504--3508},
  year={1998},
  organization={IEEE}
}

@inproceedings{kaul2012real,
  title={Real-time status: How often should one update?},
  author={Kaul, Sanjit and Yates, Roy and Gruteser, Marco},
  booktitle={2012 Proceedings IEEE INFOCOM},
  pages={2731--2735},
  year={2012},
  organization={IEEE}
}

@article{yates2021age,
  title={Age of information: An introduction and survey},
  author={Yates, Roy D and Sun, Yin and Brown, D Richard and Kaul, Sanjit K and Modiano, Eytan and Ulukus, Sennur},
  journal={IEEE Journal on Selected Areas in Communications},
  volume={39},
  number={5},
  pages={1183--1210},
  year={2021},
  publisher={IEEE}
}

@article{bhat2000finite,
  title={Finite-time stability of continuous autonomous systems},
  author={Bhat, Sanjay P. and Bernstein, Dennis S.},
  journal={SIAM Journal on Control and Optimization},
  volume={38},
  number={3},
  pages={751--766},
  year={2000},
  publisher={Society for Industrial and Applied Mathematics}
}

@inproceedings{ames2019control,
  title={Control barrier functions: Theory and applications},
  author={Ames, Aaron D and Coogan, Samuel and Egerstedt, Magnus and Notomista, Gennaro and Sreenath, Koushil and Tabuada, Paulo},
  booktitle={2019 18th European control conference (ECC)},
  pages={3420--3431},
  year={2019},
  organization={Ieee}
}

@book{alur2015principles,
  title={Principles of cyber-physical systems},
  author={Alur, Rajeev},
  year={2015},
  publisher={MIT press}
}

\end{document}